\documentclass[11pt]{article}
\pdfoutput=1
\usepackage[margin=1in]{geometry}

\usepackage{dsfont}
\usepackage{amsmath,amsfonts,amssymb,amsthm}
\usepackage{thm-restate}
\usepackage{mathtools}
\usepackage{enumitem}
\usepackage{microtype}
\usepackage{xcolor}
\usepackage{comment}
\usepackage{braket} 

\usepackage[nobiblatex]{xurl} 
\usepackage[pagebackref]{hyperref}
\hypersetup{
    pdftitle={}, 
    pdfauthor={}, 
    colorlinks=true, 
    linkcolor=blue, 
    citecolor=blue, 
    urlcolor=blue 
}

\renewcommand{\backref}[1]{}

\renewcommand{\backrefalt}[4]{%
\ifcase #1 %
\or
[p.\ #2]%
\else
[pp.\ #2]%
\fi}

\makeatletter
\newcommand{\para}{%
  \@startsection{paragraph}{4}%
  {\z@}{2ex \@plus 3.3ex \@minus .2ex}{-1em}%
  {\normalfont\normalsize\bfseries}%
}
\makeatother

\usepackage[nameinlink]{cleveref}
\crefname{equation}{Eq.}{Eqs.}
\Crefname{equation}{Equation}{Equations}

\newtheorem{theorem}{Theorem}
\newtheorem{lemma}[theorem]{Lemma}
\newtheorem{proposition}[theorem]{Proposition}

\theoremstyle{definition}

\newcommand{\abs}[1]{| #1 |}

\newcommand{\defeq}{\coloneqq}

\newcommand{\XOR}{\mathrm{XOR}}
\newcommand{\LR}{\mathrm{LR}}
\newcommand{\eps}{\epsilon}
\newcommand{\Dtr}{D_{\mathrm{tr}}}

\newcommand{\id}{\mathds{1}} 

\newcommand{\norm}[1]{\|{#1}\|}

\DeclareMathOperator{\polylog}{polylog}

\begin{document}

\title{An Optimal Quantum Linear Systems Algorithm}

\author{
Carlos Bravo-Prieto\footnote{Dahlem Center for Complex Quantum Systems,
Freie Universität Berlin.}
\and
Aram W. Harrow\footnote{Center for Theoretical Physics -- a Leinweber Institute, MIT.} 
\and
Robin Kothari\footnote{Google Quantum AI. Some of this work was done when the author was a postdoctoral associate at MIT.}
}
\date{\vspace{-2em}}

\maketitle

\begin{abstract}
In the quantum linear systems problem (QLSP), we are given query access to a $d$-sparse $N\times N$ matrix $A$ with condition number $\kappa$, and the ability to prepare a quantum state proportional to a vector $\vec b$. The goal is to output an $\eps$-approximation to the quantum state proportional to the solution $\vec x$ of $A\vec{x}=\vec{b}$. 
Following a long line of work, the best previously known quantum algorithms for the QLSP had query complexities $O(\kappa d\log(1/\eps))$~\cite{CAS21} and $\kappa\sqrt d(\kappa d/\eps)^{o(1)}$~\cite{Low19}, while the best known lower bounds were $\Omega(\kappa\log(1/\eps))$~\cite{CABB25} and $\Omega(\kappa\sqrt d)$~\cite{MKBR26}. 

We improve these bounds and show that the complexity of the QLSP is
$\Theta(\kappa\sqrt d\log(1/\eps))$.  
We also resolve an open problem of Berry and Childs~\cite{BC12} by showing that any $N\times N$ unitary can be implemented with bounded error using $O(\sqrt N)$ queries to its matrix entries.
\end{abstract}


\section{Introduction}
\label{sec:intro}

The quantum linear systems problem (QLSP) is a central algorithmic primitive in quantum computing, and was introduced by Harrow, Hassidim, and Lloyd~\cite{HHL09}. 
Informally, it is the problem of solving a linear system of equations $A \vec{x} = \vec{b}$, but where we are given quantum access to the inputs $A$ and $b$, and only asked to approximately prepare a quantum state proportional to $\vec{x}$.

\para{Problem statement.} Let $A$ be an invertible $d$-sparse $N\times N$ matrix satisfying $\norm{A}\leq1$ and $\norm{A^{-1}}\leq\kappa$, where $\kappa$ is known. Let $\vec b=(b_1,\ldots,b_N)\neq0$ and $\vec x=(x_1,\ldots,x_N)=A^{-1}\vec b$, and define
\begin{equation}
  \ket b
  \defeq
  \frac{\sum_{i=1}^N b_i\ket i}
       {\norm{\sum_{i=1}^N b_i\ket i}},
  \qquad
  \ket x
  \defeq
  \frac{\sum_{i=1}^N x_i\ket i}
       {\norm{\sum_{i=1}^N x_i\ket i}}.
\end{equation}
The goal is to output a quantum state together with a flag indicating success, with success probability at least $1/2$. Conditioned on success, the output state must have trace distance at most $\eps$ from $\ket x$.

The matrix $A$ is accessed through the standard sparse-matrix oracles: a location oracle and a value oracle. On input a row index $j\in[N]$ and an index $\ell\in[d]$, the location oracle maps $\ket{j,\ell}$ to $\ket{j,\nu(j,\ell)}$, returning the position $\nu(j,\ell)$ of the $\ell$th potentially nonzero entry in row $j$. It must list every nonzero entry in the row, but it may also list zero entries because $d$ is only an upper bound on the sparsity. The value oracle maps $\ket{j,k,z}$ to $\ket{j,k,z\oplus A_{jk}}$. We assume access to the corresponding oracles for both $A$ and $A^\dagger$. The state $\ket b$ is supplied by an efficient state-preparation circuit, and query complexity counts only queries to the matrix oracles.

Here we have used the notation and access model of Childs, Kothari, and Somma~\cite{CKS17}, and refer the reader to it for a discussion of why the problem is formulated in this way.

\para{Prior work.} The query complexity of the QLSP is governed by three parameters: the condition number $\kappa$, the sparsity $d$, and the target accuracy $\eps$~\cite{reviewQLSP}. Harrow, Hassidim, and Lloyd gave a quantum algorithm for QLSP with query complexity $\widetilde{O}(\kappa^2/\eps)$ with polynomial scaling in $d$~\cite{HHL09}. The $\kappa$ dependence of this algorithm was then improved to nearly linear by Ambainis~\cite{Amb12}, and the $\eps$ dependence was later improved to $\polylog(1/\eps)$ by Childs, Kothari, and Somma~\cite{CKS17} resulting in an algorithm with complexity $O(\kappa d \polylog(d\kappa/\eps))$.

Later, several algorithms were proposed using ideas from adiabatic quantum computing~\cite{SSO19,AL19,LT20}, culminating in the result of Costa, An, Sanders, Su, Babbush, and Berry~\cite{CAS21} with query complexity 
\begin{equation}
O(\kappa d \log(1/\eps)).    
\end{equation}
Using an improved Hamiltonian simulation algorithm, Low~\cite{Low19} presented a QLSP algorithm with query complexity 
\begin{equation}
\kappa \sqrt{d} (\kappa d/\eps)^{o(1)}.    
\end{equation}
These are incomparable, since the latter algorithm has better $d$ dependence but the former algorithm has better $\kappa$ and $\eps$ dependence. These are the best known quantum algorithms for the QLSP. 

Prior lower bounds also captured only part of the story. The known lower bounds were
\begin{equation}
    \Omega(\kappa \log(1/\eps))
\end{equation}
by Costa, An, Babbush, and Berry~\cite{CABB25}, with a joint dependence on the condition number and precision, and 
\begin{equation}
  \Omega(\kappa\sqrt{d})  
\end{equation}
by  Mori, Kikuchi, Benedetti, and Rosenkranz~\cite{MKBR26}, with constant error.
As usual, all multiparameter lower bounds are understood to hold for sufficiently large values of $\kappa$, $d$, and $1/\eps$.

\para{Main results.} We close this gap by proving matching upper and lower bounds.

\begin{theorem}[Optimal query complexity of the QLSP]
\label{thm:main}
The query complexity of the QLSP is
\begin{equation}
  \Theta(\kappa\sqrt d\log(1/\eps)).
\end{equation}
\end{theorem}

We prove the upper bound in \Cref{sec:upper-bound} with a matching lower bound in \Cref{sec:lower_bound}.

\para{Bonus results.}
We also obtain an improved block encoding of a sparse matrix itself.
A block encoding represents a matrix as a block of a larger unitary; the parameters $(\alpha,\eta)$ specify its normalization and operator-norm error, respectively.

\begin{restatable}[Square-root-sparsity block encoding]{theorem}{sqrtdblockencoding}
\label{thm:intro-sqrt-d-block-encoding}
\label{thm:sqrt-d-block-encoding}
Let $H$ be a Hermitian $d$-sparse matrix with
$\norm H\leq1$.  For every $\eta>0$, we can construct an $(O(1),\eta)$ block encoding of $H$ using
$O(\sqrt d\log(1/\eta))$ sparse-matrix oracle queries.
\end{restatable}

In the black-box unitary implementation problem, we are given
oracle access to the entries of an $N\times N$ unitary $U$ and seek a
quantum circuit that approximates $U$. For constant error, Berry and
Childs~\cite{BC12} gave an algorithm using
$O(N^{2/3}(\log\log N)^{4/3})$ queries and asked whether the
$\Omega(\sqrt N)$ lower bound from unstructured search could be matched.
Low~\cite{Low19} improved the upper bound to $O(N^{1/2+o(1)})$.
Applying \Cref{thm:sqrt-d-block-encoding} to the Hermitian dilation of
$U$, we obtain the optimal $O(\sqrt N)$ bound for constant error,
resolving this open question.

\begin{restatable}[Black-box unitary implementation]{corollary}{blackboxunitary}
\label{cor:black-box-unitary-implementation}
Given oracle access to the entries of an $N\times N$ unitary $U$, and any $\eps>0$, we can implement $U$ to error at most $\eps$ using
$O(\sqrt N\log(1/\eps))$ queries.
\end{restatable}

We prove \Cref{thm:sqrt-d-block-encoding} and
\Cref{cor:black-box-unitary-implementation} in \Cref{sec:block-encoding-A}.

\subsection{Overview of proof techniques} 

\para{Upper bound.} We first construct an exact block encoding of a
larger matrix $M$ using $O(1)$ queries to the sparse-matrix oracles for
$A$. The inverse $M^{-1}$ contains a scaled copy of $A^{-1}$ and has
norm $O(\kappa\sqrt d)$ (\Cref{thm:sparse-reduction}). We then solve the enlarged linear system using
the block-encoded linear-systems algorithm of Costa et al.~\cite{CAS21}
(\Cref{thm:costa-algorithm})
and recover the original solution with constant probability. This gives
the optimal $O(\kappa\sqrt d\log(1/\eps))$ query complexity.

To construct $M$, we start from the factorization underlying the
standard sparse-matrix block encoding: up to normalization,
$A=F^\dagger T$~\cite{GSLW19}.
Here $T$ depends only on the sparsity structure of $A$, as provided by
the location oracle, whereas $F$ uses both the locations and values of
the nonzero entries of $A$. Implementing $T$ requires only $O(1)$ queries,
but implementing $F$ can require $\Omega(\sqrt d)$
queries. Paying this cost loses the desired improvement in sparsity.

We use a simple linear-algebraic idea: an equation
involving a product of two matrices can be split into two equations by
introducing an intermediate vector $\vec y$:
\begin{equation}
  CD\vec x=\vec b
  \quad\Longleftrightarrow\quad
  \begin{aligned}
    C\vec y&=\vec b,\\
    D\vec x&=\vec y.
  \end{aligned}
  \label{eq:intro-split-product}
\end{equation}
Applying this to $F^\dagger T$
gives a larger linear system that contains the original solution among
its variables. Alternatively, this construction can be viewed as
replacing every original variable by $d$ equally weighted copies, one
for each listed equation in which it may appear.

We then block encode the enlarged system directly. The natural
state-preparation maps now load individual entries of $A$, rather than
entire rows. Normalizing each prepared state separately produces the
rescaled matrix $M$ and its exact block encoding using $O(1)$ queries.

\para{Lower bound.} We first turn a quantum circuit into a linear
system whose solution is a weighted superposition of its intermediate
states (\Cref{lem:circuit-to-matrix}), following the history-state approach to QLSP lower
bounds~\cite{HHL09}. This lets us transfer query lower bounds for
computing the circuit's output to the QLSP.

Two warm-ups explain the sources of hardness. First, embedding an
unstructured search problem gives an $\Omega(\sqrt d)$ lower bound,
even when the condition number and error are constant. Second,
embedding a circuit that computes parity gives an
$\Omega(\kappa\log(1/\eps))$ lower bound, even for constant sparsity.
The condition number controls how quickly the weights in the history
state decay, while the required precision determines how late in the
computation a signal remains detectable. A parity circuit of length
proportional to $\kappa\log(1/\eps)$ still leaves enough output signal
to survive the allowed error, yielding this second lower bound.

To obtain the joint bound, we combine the two constructions: compute
the parity of $R=\Theta(\kappa\log(1/\eps))$ bits, each coming from an
independent search problem of size $\Theta(d)$. The remaining issue is
that reading the answer from the history state may give only a small
advantage over random guessing. The quantum XOR lemma of Lee and
Roland~\cite{LR12} shows that the composed problem remains hard even
with this small advantage, requiring $\Omega(R\sqrt d)$ queries.
Thus the two sources of hardness multiply, giving the desired
$\Omega(\kappa\sqrt d\log(1/\eps))$ lower bound.

\para{Block encoding and unitary implementation.} First suppose that
the input matrix $A$ is well-conditioned. Starting from the block
encoding of $M$ provided by our upper-bound reduction,
we apply standard block-encoding
inversion~\cite{GSLW19} to block encode $M^{-1}$. Since $M^{-1}$
contains a scaled copy of $A^{-1}$, selecting this block gives a
constant-normalization block encoding of $A^{-1}$ using $O(\sqrt d)$
queries at constant error (\Cref{lem:sparse-inverse-block-encoding}).

Now let $H$ be an arbitrary Hermitian $d$-sparse matrix with
$\norm H\leq1$. Unlike $A$, it need not be well-conditioned or even
invertible. We construct a larger sparse matrix $G_H$ that is
well-conditioned and whose inverse contains $H$ as a block. Applying
the preceding inverse-encoding procedure to $G_H$ therefore lets us
block encode $H$ itself.

Finally, for a unitary $U$, its Hermitian dilation $H_U$ is both
well-conditioned and its own inverse. Applying the inverse-encoding
procedure to $H_U$ and selecting the $U$ block gives a block encoding
of $U$. Robust oblivious amplitude amplification~\cite{BCK15} then
implements $U$ with bounded error using the optimal $O(\sqrt N)$
entry queries.

\subsection*{Concurrent Work}
In concurrent work, Patel~\cite{Pat26} also obtains the lower bound
$\Omega(\kappa\sqrt d\log(1/\eps))$ for the quantum linear systems
problem in the sparse-matrix access model.

\subsection*{Acknowledgements}
C.B.-P. thanks Antonio Anna Mele for insightful discussions on open problems in quantum linear systems.
R.K. thanks Suhail Sherif and Praneeth Netrapalli for many interesting discussions about variants and special cases of the quantum linear systems problem.
A.W.H. was supported by the U.S. Department of Energy, Office of Science, National Quantum Information Science Research Centers, Quantum Systems Accelerator (Award No. DE-SCL0000121).
C.B.-P. acknowledges the Munich Quantum Valley for financial support. ChatGPT Pro 5.6 contributed to the proof development of the upper bound in \Cref{thm:upper-bound} and various versions of ChatGPT, Claude, and Gemini were used to simplify the proof and prepare this manuscript.
The remaining proofs were developed independently by the authors. The authors take full responsibility for the correctness of the results.

\section{Quantum algorithm for the QLSP}
\label{sec:upper-bound}

We now prove the upper bound in \Cref{thm:main}.

\begin{theorem}[Quantum algorithm]
\label{thm:upper-bound}
There is a quantum algorithm that solves the QLSP using
\begin{equation}
  O(\kappa\sqrt d\log(1/\eps)) \ \text{queries.}
  \label{eq:upper-bound-complexity}
\end{equation}
\end{theorem}

To explain the proof structure, we need to recall the concept of a block encoding.
A unitary $U_K$ is an (exact) block encoding of $K$ if, for all
computational-basis states $\ket i,\ket j$ of the system register,
\begin{equation}
  (\bra{0^a}\otimes\bra i)U_K
  (\ket{0^a}\otimes\ket j)
  =K_{ij}.
  \label{eq:block-encoding-definition}
\end{equation}

By the standard Hermitian embedding~\cite{HHL09}, we may assume without loss of generality that $A=A^\dagger$.
To prove \Cref{thm:upper-bound}, we show how to use sparse-matrix oracle access to $A$ to get an exact block encoding of a larger matrix $M$ using $O(1)$ queries. We then solve the larger linear system using $M$ and show that the solution to this allows us to recover the solution to our original linear system. The following theorem summarizes this reduction.

\begin{theorem}[Sparse-matrix oracle QLSP reduction]
\label{thm:sparse-reduction}
Let $A$ be an invertible Hermitian $d$-sparse matrix
with $\norm A\leq1$ and $\norm{A^{-1}}\leq\kappa$.  Given standard
sparse-matrix oracle access to $A$, one can implement an exact block encoding of an invertible matrix $M$ using $O(1)$ queries. $M$ has the
following properties:
\begin{enumerate}
  \item $M$ is a contraction (i.e., $\norm{M} \leq 1$) and $\norm{M^{-1}} \leq 5\kappa\sqrt d$.
  \item The bottom-right block of $M^{-1}$ is $\sqrt d\,A^{-1}$.
  \item For every nonzero $\vec b$, measuring the bottom component of the
  normalized state proportional to $M^{-1}\binom{0}{\vec b}$ yields
  $A^{-1}\vec b/\norm{A^{-1}\vec b}$ with constant probability.
\end{enumerate}
\end{theorem}

To solve the enlarged system, we use the QLSP algorithm of Costa et al.~\cite{CAS21}, which assumes block-encoding access to the input matrix.

\begin{theorem}[QLSP from block encoding~{\cite[Theorem~11]{CAS21}}] \label{thm:costa-algorithm}
Let $A$ be an invertible matrix with $\norm A\leq1$ and $\norm{A^{-1}}\leq\kappa$. Given an oracle block encoding the matrix $A$ and an oracle preparing $\ket{b}$, there exists a quantum algorithm that produces the normalized state $\ket{x}$ to within error $\eps$ with success probability at least $1/2$, using
\begin{equation}
  O(\kappa\log(1/\eps)) \ \text{queries.}\footnote{There is no dependence on $d$ in this algorithm because it assumes access to a block encoding of $A$, not to a sparse-matrix oracle for $A$.}
\end{equation}
\end{theorem}

Combining these two theorems completes the proof of \Cref{thm:upper-bound}. The rest of this section is devoted to proving \Cref{thm:sparse-reduction}. The reduction in \Cref{thm:sparse-reduction} also yields an efficient block encoding of $A$ itself, as shown in \Cref{thm:sqrt-d-block-encoding}.

\subsection{The standard sparse-matrix oracle to block-encoding reduction}
\label{sec:old-overlap}

Fix $n\geq1$, set $N=2^n$, and write $[N]=\{1,\ldots,N\}$. For each $i\in[N]$, the location
oracle lists $d$ distinct positions
\begin{equation}
  \mathcal N(i)
  =\{\nu(i,1),\ldots,\nu(i,d)\}
  \label{eq:neighbor-lists}
\end{equation}
that contain the support of row $i$, where some of the listed positions may have value
zero.  The oracle implements the reversible map
$\ket{i,\ell}\mapsto\ket{i,\nu(i,\ell)}$.  Since $A$ is Hermitian,
$\mathcal N(i)$ also contains the support of column $i$.  A value oracle
returns $A_{ij}$.

We first review the sparse-matrix oracle block encoding of Gily\'en et
al.~\cite{GSLW19} in a form suited to our construction.  
The construction uses two maps:
\begin{equation}
  T\ket j
  \defeq
  \Bigl(\frac1{\sqrt d}\sum_{i\in\mathcal N(j)}\ket i\Bigr)\otimes\ket j,
  \qquad
  F\ket i
  \defeq
  \ket i\otimes
  \Bigl(\sum_{j\in\mathcal N(i)}\overline{A_{ij}}\ket j\Bigr).
  \label{eq:T-F-definition}
\end{equation}
The map $T$ retains its input in the second register and distributes it
uniformly among its $d$ listed equation slots, which include every
equation in which the variable occurs. The map $F$ retains its input in the first
register and loads the conjugated row of $A$ into its listed slots.
For both maps, distinct basis inputs produce orthogonal outputs, and
\begin{equation}
  T^\dagger T=\id,
  \qquad
  F^\dagger F
  =
  \sum_{i=1}^N
  \left(\sum_{j=1}^N\abs{A_{ij}}^2\right)
  \ket i\!\bra i
  \ \preceq \ \norm{A} \sum_{i=1}^N \ket i\!\bra i
  \ \preceq\ \id.
  \label{eq:T-F-norms}
\end{equation}
Hence $T$ is an isometry and $F$ is a contraction.  Equivalently, $F$
can be completed to an isometry by adding a flag register.

Their overlap has the desired form:
\begin{equation}
  \bra iF^\dagger T\ket j
  =\braket{Fi|Tj}
  =\frac{A_{ij}}{\sqrt d},
  \qquad\text{and hence}\qquad
  F^\dagger T=\frac A{\sqrt d}.
  \label{eq:dream-overlap}
\end{equation}
Here Hermiticity ensures that, if $A_{ij}\neq0$, then
$i\in\mathcal N(j)$ and $j\in\mathcal N(i)$.  Listed zero entries do not
contribute to the overlap.

To obtain the usual block encoding, use $T$ as one of the two state preparation isometries and an isometric completion of $F/\sqrt{d}$ as the other.  The former requires
only a uniform superposition over the listed positions.  For the
latter, a location query followed by a value query and a controlled
rotation implements
\begin{equation}
  \ket i
  \longmapsto
  \ket i \otimes \frac1{\sqrt d} 
  \sum_{j\in\mathcal N(i)}
  \ket j
  \left(
    \overline{A_{ij}}\ket0
    +\sqrt{1-\abs{A_{ij}}^2}\ket1
  \right).
  \label{eq:old-row-preparation}
\end{equation}
Its flag-$0$ component is exactly $F\ket i/\sqrt d$.  The
overlap of these two isometries is therefore
\begin{equation}
  \left(\frac F{\sqrt d}\right)^\dagger T
  =\frac Ad
  \label{eq:old-overlap}
\end{equation}
which gives an exact block encoding of $A/d$.

The inverse norm seen by the block-encoded linear-systems algorithm is
\begin{equation}
  \norm{(A/d)^{-1}}
  =d\norm{A^{-1}}
  \leq\kappa d,
  \label{eq:old-effective-condition}
\end{equation}
so the standard construction uses
$O(\kappa d\log(1/\eps))$ sparse-matrix oracle queries.  By contrast, an
$O(1)$-query block encoding of $A/\sqrt d$ would have inverse norm at
most $\kappa\sqrt d$ and would exactly provide the desired improvement.

The two factors $1/\sqrt d$ in \cref{eq:old-overlap} have different
origins.  The factor in $T$ is the unavoidable cost of distributing one
amplitude over $d$ orthogonal slots.  The other arises from preparing
the entire row $F\ket i$ with amplitude scale $1/\sqrt d$.  Preparing
$F\ket i$ itself would make its overlap with $T$ equal to
$A/\sqrt d$, as required.

There is no linear-algebraic obstruction: $F$ is a contraction.  The
obstruction is algorithmic. Suppose one entry in row $i$ has magnitude
nearly one at an unknown position $j^*$, while all other entries are
arbitrarily small. Measuring a state close to $F\ket i$ would then reveal
$j^*$, hence solving an instance of unstructured search, which requires
$\Omega(\sqrt d)$ queries~\cite{BBBV97}.  Paying this cost would
restore the total complexity to $O(\kappa d\log(1/\eps))$. Therefore, although $F$ is the desired map, we avoid implementing it directly.

\subsection{A larger linear system}
\label{sec:augmented-system}

We now show how to use $F^\dagger T=A/\sqrt d$ without implementing $F$.
The difficulty is that one state would have to contain an entire row of $A$.
Instead, assign a separate variable to each pair $(i,j)$, representing variable $j$ in equation $i$, and let the linear-system solver perform the sum over $j$.

The argument in this section is purely ordinary linear algebra, and none of its vectors is
assumed normalized.  Let $\vec b$ be the right-hand side and let
$\vec x=A^{-1}\vec b$ be the solution of $A\vec x=\vec b$.  Since
$F^\dagger T=A/\sqrt d$, define
\begin{equation}
  \vec u
  \defeq
  \sqrt d\,\vec x
  =\sqrt d\,A^{-1}\vec b.
  \label{eq:u-and-x}
\end{equation}
Thus $\vec u$ has the same direction as the desired
solution $\vec x$, and it is the unique solution of
\begin{equation}
  F^\dagger T\vec u=\vec b.
  \label{eq:product-system}
\end{equation}
Introduce the intermediate vector
\begin{equation}
  \vec z=T\vec u,
  \qquad
  F^\dagger\vec z=\vec b.
  \label{eq:intermediate-equations}
\end{equation}
In coordinates, these equations read
\begin{equation}
  z_{ij}
  =
  \begin{cases}
    u_j/\sqrt d,&i\in\mathcal N(j),\\
    0,&i\notin\mathcal N(j),
  \end{cases}
  \qquad
  \sum_{j=1}^N A_{ij}z_{ij}=b_i.
  \label{eq:intermediate-coordinates}
\end{equation}
Thus one variable $u_j$ has been replaced by $d$ listed copies, with all other coordinates set to zero.  The first equations impose
$\vec z=T\vec u$, and the second impose $F^\dagger\vec z=\vec b$.
Together they give $A\vec u/\sqrt d=\vec b$.  This is the standard
auxiliary-variable trick for replacing a product by two equations.

We can combine the two equations in \cref{eq:intermediate-equations} into a single linear system:
\begin{equation}
  \underbrace{
  \begin{pmatrix}
    \id/\sqrt2&-T/\sqrt2\\
    F^\dagger&0
  \end{pmatrix}}_{=\,L}
  \begin{pmatrix}
    \vec z\\
    \vec u
  \end{pmatrix}
  =
  \begin{pmatrix}
    0\\
    \vec b
  \end{pmatrix}.
  \label{eq:M-system}
\end{equation}
The factor $1/\sqrt2$ has no algebraic significance.  It normalizes
the two coefficients in each equation $\vec z=T\vec u$, which is
convenient for the subsequent block encoding. Note that the number of variables has increased from $N$ to $N^2+N$.  This
would generally be expensive for a classical linear-system algorithm,
but a quantum register holding $N^2+N\leq2N^2$ basis states requires only
$2n+1$ qubits.

Now let us upper bound the condition number of $L$. 

\begin{proposition}
The matrix $L$ is invertible and satisfies $\norm{L}\norm{L^{-1}} = O(\kappa \sqrt{d})$.
\end{proposition}
\begin{proof}
Since the 4 blocks in $L$ are all contractions, $\norm{L}=O(1)$.
Using $F^\dagger T=A/\sqrt d$, we have
\begin{equation}
  \underbrace{
  \begin{pmatrix}
    \sqrt2\bigl(\id-\sqrt d\,TA^{-1}F^\dagger\bigr)
      &\sqrt d\,TA^{-1}\\
    -\sqrt{2d}\,A^{-1}F^\dagger
      &\sqrt d\,A^{-1}
  \end{pmatrix}}_{=\,L^{-1}}
  \underbrace{
  \begin{pmatrix}
    \id/\sqrt2&-T/\sqrt2\\
    F^\dagger&0
  \end{pmatrix}}_{=\,L}
  =
  \begin{pmatrix}
    \id&0\\
    0&\id
  \end{pmatrix}.
  \label{eq:L-back-substitution}
\end{equation}
Since  $\norm T=1$, $\norm{F^\dagger}\leq1$, and
$\norm{A^{-1}}\leq\kappa$, each block of $L^{-1}$
has norm $O(\kappa\sqrt d)$, and hence
$\norm{L^{-1}}=O(\kappa\sqrt d)$.
\end{proof}

Thus $L$ already has the desired condition-number
scale.  A block encoding of $L/\sqrt2$ using $O(1)$ queries would complete the construction. However, $L$ contains $F^\dagger$, whose direct
implementation is precisely the obstruction identified in \Cref{sec:old-overlap}.  It remains to modify $L$ so that it
can be block encoded without spoiling its inverse bound or its
solution.

\subsection{Block encoding the larger linear system}
\label{sec:normalization}

As in \Cref{sec:old-overlap}, we seek a block encoding by writing a suitably normalized version of $L$ as $V_L^\dagger V_R$, where $V_L$ and $V_R$ are contractions, completed to isometries when needed. The normalized maps will have the form
\begin{equation}
  \begin{aligned}
    V_L&=
    \begin{pmatrix}
      \ket{\ell_{11}}&\ket{\ell_{12}}&\cdots&\ket{\ell_{NN}}
      &\ket{\rho_1}&\cdots&\ket{\rho_N}
    \end{pmatrix},\\
    V_R&=
    \begin{pmatrix}
      \ket{\psi_{11}}&\ket{\psi_{12}}&\cdots&\ket{\psi_{NN}}
      &\ket{\tau_1}&\cdots&\ket{\tau_N}
    \end{pmatrix}.
  \end{aligned}
  \label{eq:VL-VR-columns}
\end{equation}
We first choose vectors whose inner products reproduce $L$. The $\psi$ states will initially be unnormalized, so we denote them by $\ket{\widetilde\psi_{ij}}$.

We start with the $\tau$ states, which correspond to the second block column of $L$. Since this column contains $-T/\sqrt2$, a natural choice is to include $T\ket j$ in $\ket{\tau_j}$. Introduce mutually orthogonal auxiliary labels $\mathsf z$, $\mathsf u$, and $\mathsf r$, and set
\begin{equation}
  \ket{\tau_j}
  \defeq
  \ket{\mathsf u}\otimes T\ket j
  =
  \ket{\mathsf u}
  \Bigl(\frac1{\sqrt d}
  \sum_{i'\in\mathcal N(j)}
  \ket{i'}\Bigr)\ket j.
  \label{eq:tau-state}
\end{equation}
To obtain the required inner products with the $\tau$ states, we give the $\ell$ states a $-\ket{\mathsf u}/\sqrt2$ component, together with a $\ket{\mathsf z}/\sqrt2$ component for the identity block. The $\rho$ states use the $\mathsf r$ label to remain orthogonal to the $\tau$ states, and a uniform superposition over the listed positions in row $i$. We choose uniform amplitudes so that the $\rho$ states require only a location query; the coefficients $A_{ij}$ will enter through the $\widetilde\psi$ states. Thus define
\begin{equation}
  \ket{\ell_{ij}}
  \defeq
  \frac{
    \ket{\mathsf z}-\ket{\mathsf u}
  }{\sqrt2}\ket{i,j},
  \qquad
  \ket{\rho_i}
  \defeq
  \ket{\mathsf r}\otimes\mathrm{SWAP}\,T\ket i
  =
  \ket{\mathsf r}\ket i
  \Bigl(\frac1{\sqrt d}
  \sum_{j'\in\mathcal N(i)}
  \ket{j'}\Bigr).
  \label{eq:left-states-preview}
\end{equation}
Here $\mathrm{SWAP}$ exchanges the two index registers.
The $\ell$, $\rho$, and $\tau$ states are normalized. It remains to choose the $\widetilde\psi$ states. The $\mathsf z$ amplitude of $\ket{\widetilde\psi_{ij}}$ should be $1$, so that its inner product with $\ket{\ell_{ij}}$ is $1/\sqrt2$. Its $\mathsf r$ amplitude should be $\sqrt d\,A_{ij}$ to compensate for the factor $1/\sqrt d$ in $\ket{\rho_i}$. Thus define
\begin{equation}
  \ket{\widetilde\psi_{ij}}
  \defeq
  \bigl(\ket{\mathsf z}
  +\sqrt d\,A_{ij}\ket{\mathsf r}\bigr)\ket{i,j}.
  \label{eq:raw-z-state}
\end{equation}
These choices give the required inner products:
\begin{align}
  \braket{\ell_{km}|\widetilde\psi_{ij}}
  &=\frac{\delta_{ki}\delta_{mj}}{\sqrt2},
    \notag
  &\braket{\ell_{km}|\tau_j}
  &=-\frac{
    \delta_{mj}\,\mathbf 1_{\{k\in\mathcal N(j)\}}
  }{\sqrt{2d}},
  \\
  \braket{\rho_k|\widetilde\psi_{ij}}
  &=\delta_{ki}A_{ij},
  &
  \braket{\rho_k|\tau_j}
  &=0.
  \label{eq:raw-overlaps}
\end{align}
With the column ordering in \cref{eq:VL-VR-columns}, the product $V_L^\dagger V_R$ before normalization is
\begin{equation}
\begin{array}{c@{\;\;}c}
  \begin{array}{c|cc}
    & \{\ket{\widetilde\psi_{km}}\}_{k,m}
    & \{\ket{\tau_j}\}_j \\ \hline
    \{\bra{\ell_{ij}}\}_{i,j}
    & \id/\sqrt2 & -T/\sqrt2 \\
    \{\bra{\rho_i}\}_i
    & F^\dagger & 0
  \end{array}
  &
  \raisebox{-1.2ex}{$=L$}
\end{array}
\label{eq:raw-overlap-matrix}
\end{equation}

\para{Normalization.}
The $\ell$ and $\rho$ states together form an orthonormal family, so $V_L$ is already an isometry. The $\widetilde\psi$ and $\tau$ states are also mutually orthogonal, but the $\widetilde\psi$ states need not have unit norm. Indeed, the two terms in \cref{eq:raw-z-state} are orthogonal, giving
\begin{equation}
  q_{ij}
  \defeq
  \norm{\widetilde\psi_{ij}}
  =
  \sqrt{1+d\abs{A_{ij}}^2},
  \qquad
  1\leq q_{ij}\leq\sqrt{d+1}.
  \label{eq:q-definition}
\end{equation}
It remains to scale the $\widetilde\psi$ states to have norm at most $1$, so that $V_R$ is a contraction. Dividing them all by $\sqrt{d+1}$ would suffice, but imposes the worst-case scaling at every pair $(i,j)$, regardless of the size of $A_{ij}$.

The enlarged system allows each vector $\ket{\widetilde\psi_{ij}}$ to be normalized separately, which is not possible in the original construction. Instead, normalize \cref{eq:raw-z-state} by its own norm. That is, prepare
\begin{equation}
  \ket{\psi_{ij}}
  \defeq
  \frac{\ket{\widetilde\psi_{ij}}}{q_{ij}}
  =
  \frac{
    \ket{\mathsf z}
    +\sqrt d\,A_{ij}\ket{\mathsf r}
  }{q_{ij}}\ket{i,j}.
  \label{eq:local-z-state}
\end{equation}
This is a unit vector for every $i,j$. Let
\begin{equation}
  Q\ket{i,j}
  \defeq
  q_{ij}\ket{i,j}.
  \label{eq:Q-definition}
\end{equation}
Dividing each $\ket{\widetilde\psi_{ij}}$ by $q_{ij}$ divides the corresponding column of $L$ by $q_{ij}$, while the $\tau$ columns remain unchanged. Thus $V_L^\dagger V_R=M$, where
\begin{equation}
  M
  \defeq
  L
  \begin{pmatrix}
    Q^{-1}&0\\
    0&\id
  \end{pmatrix}
  =
  \begin{pmatrix}
    Q^{-1}/\sqrt2&-T/\sqrt2\\
    F^\dagger Q^{-1}&0
  \end{pmatrix}.
  \label{eq:Mhat-definition}
\end{equation}

Preparing $\ket{\ell_{ij}}$ requires no oracle queries, while
$\ket{\rho_i}$ and $\ket{\tau_j}$ each require one location query.
Preparing $\ket{\psi_{ij}}$ requires a value query, a controlled
two-level rotation, and an inverse value query. These preparations
implement $V_L$ and $V_R$ coherently, with temporary work registers
uncomputed. Composing their unitary extensions therefore gives an exact
block encoding of $M=V_L^\dagger V_R$ with normalization factor $1$
using $O(1)$ sparse-matrix oracle queries. In particular, $\norm M\leq1$.

To use this block encoding for our original linear system in $L$, we compare
the solutions for the right-hand side $(0,\vec b)$. With
$\vec u=\sqrt d\,A^{-1}\vec b$, \cref{eq:Mhat-definition} and the
formula for $L^{-1}$ in \cref{eq:L-back-substitution} give
\begin{equation}
  L^{-1}
  \begin{pmatrix}
    0\\
    \vec b
  \end{pmatrix}
  =
  \begin{pmatrix}
    T\vec u\\
    \vec u
  \end{pmatrix}
  ,\qquad
  M^{-1}
  \begin{pmatrix}
    0\\
    \vec b
  \end{pmatrix}
  =
  \begin{pmatrix}
    QT\vec u\\
    \vec u
  \end{pmatrix}.
  \label{eq:M-L-equivalence}
\end{equation}
Thus only the top component changes, from $T\vec u$ to $QT\vec u$;
the bottom component remains
proportional to the desired solution $A^{-1}\vec b$. 

Since $M$ can be
block-encoded with $O(1)$ cost, we therefore have two remaining goals:
\begin{enumerate}
  \item Show that $\norm{QT\vec u}$ is within a constant factor of
  $\norm{\vec u}$, so that the bottom component can be recovered from
  the normalized solution with constant probability.
  \item Prove that the condition number $\norm M\norm{M^{-1}}$ is
  $O(\kappa\sqrt d)$. 
\end{enumerate}

\subsection{Recovering the solution}
\label{sec:solution-recovery}

Our goal is to show that $\norm{QT}=O(1)$. Recall that
$Q=\operatorname{diag}(q_{ij})$, where
$q_{ij}=\sqrt{1+d\abs{A_{ij}}^2}$. Equivalently,
$q_{ij}^2=1+d\abs{A_{ij}}^2$.
Although a single $q_{ij}$ may be of order $\sqrt d$, the product $QT$ has
much smaller norm, as the following lemma shows.

\begin{lemma}[The bound on $\norm{QT}$]
\label{lem:Q-bounds}
The diagonal matrix $Q$ satisfies $\norm{QT}\leq\sqrt2$.
\end{lemma}

\begin{proof}
The $j$th column of $QT$ is
\begin{equation}
  QT\ket j
  =
  \frac1{\sqrt d}
  \sum_{i\in\mathcal N(j)}q_{ij}\ket{i,j}.
  \label{eq:QT-coordinates}
\end{equation}
Different columns have disjoint support, so they are orthogonal.
Therefore the squared spectral norm is the largest squared column norm:
\begin{align}
  \norm{QT}^2
  =\max_j\norm{QT\ket j}^2
  =\max_j\frac1d\sum_{i\in\mathcal N(j)}q_{ij}^2
  =1+ \max_j \sum_{i\in\mathcal N(j)}\abs{A_{ij}}^2 
  \leq2.
  \label{eq:QT-calculation}
\end{align}
The last inequality follows because each column of $A$ has 
Euclidean norm at most $\norm A\leq1$.
\end{proof}

By the lemma, measuring the bottom component of the normalized state
proportional to $\binom{QT\vec u}{\vec u}$ from
\cref{eq:M-L-equivalence} succeeds with probability
\begin{equation}
  \frac1{1+\norm{QT}^2}
  \geq\frac13.
  \label{eq:readout-probability}
\end{equation}
Conditioned on this outcome,
the state is the desired normalized solution
$A^{-1}\vec b/\norm{A^{-1}\vec b}$.
For approximate solutions, constant-probability postselection increases trace-distance error by only a constant factor. Reducing the solver's error by a constant factor and repeating $O(1)$ times therefore gives error at most $\eps$ with success probability at least $1/2$.

\subsection{Bound on the condition number of \texorpdfstring{$M$}{M}}
\label{sec:condition-number}

Since $M=V_L^\dagger V_R$ has $\norm M\leq1$, it suffices to show
that $\norm{M^{-1}}=O(\kappa\sqrt d)$ to obtain the desired
condition-number bound. We use $\norm{QT}\leq\sqrt2$ from
\Cref{lem:Q-bounds} and the bound $\norm Q\leq\sqrt{2d}$, which
follows from
\begin{equation}
  q_{ij}^2
  =1+d\abs{A_{ij}}^2
  \leq d+1
  \leq2d.
  \label{eq:q-worst-case}
\end{equation}

\begin{proposition}[Condition number]
\label{prop:condition-number}
The matrix $M$ is invertible and satisfies
\begin{equation}
  \norm{M^{-1}}\leq5\kappa\sqrt d.
  \label{eq:Mhat-inverse-bound}
\end{equation}
Its condition number is therefore at most $5\kappa\sqrt d$.
\end{proposition}

\begin{proof}
The matrix $M$ is invertible by \cref{eq:Mhat-definition}, since
$L$ and $Q$ are invertible.
Using the formula for $L^{-1}$ in \cref{eq:L-back-substitution},
\begin{equation}
  \begin{aligned}
    M^{-1}
    &=\begin{pmatrix}Q&0\\0&\id\end{pmatrix}L^{-1}=\begin{pmatrix}Q&0\\0&\id\end{pmatrix}
    \begin{pmatrix}
      \sqrt2\bigl(\id-\sqrt d\,TA^{-1}F^\dagger\bigr)
        &\sqrt d\,TA^{-1}\\
      -\sqrt{2d}\,A^{-1}F^\dagger
        &\sqrt d\,A^{-1}
    \end{pmatrix}\\
    &=\sqrt d
    \begin{pmatrix}QT\\\id\end{pmatrix}
    A^{-1}
    \begin{pmatrix}-\sqrt2F^\dagger&\id\end{pmatrix}
    +\sqrt2\begin{pmatrix}Q&0\\0&0\end{pmatrix}.
  \end{aligned}
  \label{eq:M-inverse-factorization}
\end{equation}
The matrices $\begin{pmatrix}QT\\\id\end{pmatrix}$ and
$\begin{pmatrix}-\sqrt2F^\dagger&\id\end{pmatrix}$ each have norm at most
$\sqrt3$, since $\norm{QT}\leq\sqrt2$ and $\norm{F^\dagger}\leq1$.
Therefore
\begin{equation}
  \norm{M^{-1}}
  \leq3\sqrt d\,\norm{A^{-1}}+\sqrt2\norm Q
  \leq3\kappa\sqrt d+2\sqrt d
  \leq5\kappa\sqrt d.
  \label{eq:five-condition-bound}
\end{equation}
The last inequality uses $\kappa\geq1$.
\end{proof}

This completes the proof of \Cref{thm:sparse-reduction}.

\section{Matching lower bound for the QLSP}\label{sec:lower_bound}

In this section, we derive a lower bound that matches the query complexity of our quantum algorithm.

\begin{theorem}[Lower bound]
\label{thm:main_lower_bound}
Any quantum algorithm that solves the QLSP must make
\begin{equation}
  \Omega(\kappa\sqrt d\log(1/\eps)) \ \text{queries.}
\end{equation}
\end{theorem}

\subsection{Circuit histories as linear systems}\label{sec:circuit}

Similar to the lower bound of $\Omega(\kappa)$ established in \cite{HHL09}, we begin with the circuit-to-matrix construction that drives our result. The construction is a cyclic version of the usual history-state idea. Given a circuit, we build a unitary clock evolution that first applies the circuit, then idles while the output is stored, and finally uncomputes. The inverse of a simple linear operator built from this clock unitary prepares a geometrically weighted history state. The useful part of this history state is the idle interval, during which the data register contains the circuit output.

Let $C=G_L\cdots G_1$ be a quantum circuit on a Hilbert space $\mathcal H$, and let $\ket{\psi_0}\in\mathcal H$ be its input state. We write
\begin{equation}
  \ket{\psi_{\mathrm{out}}}
  \defeq
  G_L\cdots G_1\ket{\psi_0}.
\end{equation}
The clock register has orthonormal basis $\ket{1},\ldots,\ket{3L}$, with the convention that $\ket{3L+1}=\ket1$. Define the clock unitary
\begin{equation}
  U
  \defeq
  \sum_{\ell=1}^{L}
  \ket{\ell+1}\bra{\ell}\otimes G_\ell
  +
  \sum_{\ell=1}^{L}
  \ket{L+\ell+1}\bra{L+\ell}\otimes I
  +
  \sum_{\ell=1}^{L}
  \ket{2L+\ell+1}\bra{2L+\ell}\otimes G_{L-\ell+1}^{\dagger}.
  \label{eq:cyclic-clock-unitary}
\end{equation}
Thus the first $L$ clock steps compute $C$, the next $L$ clock steps idle, and the final $L$ clock steps uncompute $C$. 

For $\kappa>2$, set
\begin{equation}
  \alpha
  \defeq
  1-\frac{2}{\kappa},
  \qquad
  A
  \defeq
  \frac12\left(I-\alpha U\right),
  \qquad
  \ket b
  \defeq
  \ket1\ket{\psi_0}.
  \label{eq:clock-matrix-A}
\end{equation}
Thus $U$ and $A$ are square matrices with $3L\dim(\mathcal H)$ rows and the same number of columns.
The following lemma describes the properties of this construction in the form needed later.

\begin{lemma}[Circuit-to-matrix construction]
\label{lem:circuit-to-matrix}
Let $G_1,\ldots,G_L$ be unitaries on $\mathcal H$. Let $U$, $A$, and $\ket b$ be defined by \cref{eq:cyclic-clock-unitary,eq:clock-matrix-A}. Then $U$ is unitary, and
\begin{equation}
  \norm A\le 1,
  \qquad
  \norm{A^{-1}}\le \kappa.
\end{equation}
Moreover, if
\begin{equation}
  \ket x
  \defeq
  \frac{A^{-1}\ket b}{\norm{A^{-1}\ket b}},
\end{equation}
then measuring the clock register of $\ket x$ gives an outcome in $\{L+1,\ldots,2L\}$ with probability at least
\begin{equation}
  p
  \ge
  \frac{\alpha^{4L}}{3}.
  \label{eq:clock-success-main}
\end{equation}
Conditioned on this event, the data register is exactly $\ket{\psi_{\mathrm{out}}}$.
\end{lemma}

\begin{proof}
The operator $U$ advances the clock by one step and applies a unitary to the data register at each step. Distinct input clock subspaces are mapped isometrically to distinct output clock subspaces, and the cyclic convention $\ket{3L+1}=\ket1$ closes the evolution. Hence $U$ is unitary.

Since $U$ is unitary and $0<\alpha<1$,
\begin{equation}
  \norm A
  =
  \frac12\norm{I-\alpha U}
  \le
  \frac12\left(1+\alpha\right)
  \le
  1.
\end{equation}
For every vector $\ket v$,
\begin{equation}
  \norm{\left(I-\alpha U\right)\ket v}
  \ge
  \norm{\ket v}-\alpha\norm{U\ket v}
  =
  (1-\alpha)\norm{\ket v}.
\end{equation}
Thus every singular value of $I-\alpha U$ is at least $1-\alpha$, and every singular value of $A$ is at least
\begin{equation}
  \frac{1-\alpha}{2}
  =
  \frac1\kappa.
\end{equation}
Therefore $\norm{A^{-1}}\le\kappa$.

We now compute the solution state. Since $\norm{\alpha U}<1$, the Neumann series gives
\begin{equation}
  A^{-1}\ket b
  =
  2\left(I-\alpha U\right)^{-1}\ket b
  =
  2\sum_{t=0}^{\infty}\alpha^t U^t\ket b.
  \label{eq:infinite-neumann}
\end{equation}
Starting from $\ket b=\ket1\ket{\psi_0}$, the cyclic clock returns to its starting state after $3L$ steps, so $U^{3L}\ket b=\ket b$. Grouping the terms in \cref{eq:infinite-neumann} by their residue modulo $3L$ gives
\begin{equation}
  A^{-1}\ket b
  \propto
  \sum_{t=0}^{3L-1}\alpha^t U^t\ket b.
  \label{eq:finite-history-state}
\end{equation}
For $L\le t<2L$, the clock lies in the idle interval and the data register is $\ket{\psi_{\mathrm{out}}}$. The states $U^t\ket b$ for $0\le t<3L$ have distinct clock labels and are therefore mutually orthogonal. Hence the probability of observing the middle clock interval is
\begin{equation}
  p
  =
  \frac{\sum_{t=L}^{2L-1}\alpha^{2t}}
       {\sum_{t=0}^{3L-1}\alpha^{2t}}.
\end{equation}
The numerator contains $L$ terms, each at least $\alpha^{4L}$, and the denominator contains $3L$ terms, each at most $1$. Therefore
\begin{equation}
  p
  \ge
  \frac{L\alpha^{4L}}{3L}
  =
  \frac{\alpha^{4L}}3.
\end{equation}
Conditioned on a middle clock outcome, the data register is exactly $\ket{\psi_{\mathrm{out}}}$.
\end{proof}

\subsection{Warm-up lower bounds}

Before proving the full joint lower bound, we isolate the two basic sources of hardness. The first source is sparse unstructured search, which produces the factor $\sqrt d$. The second is the decay of the clock history, which couples the condition number $\kappa$ to the output precision $\eps$ and produces the factor $\kappa\log(1/\eps)$. These warm-ups are not needed as separate theorems in the final proof, but they explain why the final construction has its particular form.

\subsubsection{Sparsity from unstructured search}
\label{subsec:warmup-sparsity}

Assume first that $\kappa$ and $\eps$ are constants. Let $D$ be a power of two, and consider a hidden shift $\sigma\in\{0,\ldots,D-1\}$. Let $P_\sigma$ be the cyclic shift
\begin{equation}
  P_\sigma\ket j
  =
  \ket{j+\sigma \bmod D}.
\end{equation}
Starting from the input state $\ket0$, a single application of $P_\sigma$ gives $\ket\sigma$. Thus learning the output of this one-gate circuit is equivalent to finding the unique marked item in an unstructured search problem of size $D$, which requires $\Omega(\sqrt D)$ queries~\cite{BBBV97}.

In the standard sparse-access model, for each row, the location oracle lists all $D$ columns. Thus the location oracle contains no information about $\sigma$. The input dependence is in the value oracle. If $y\in\{0,1\}^D$ has Hamming weight one and its unique $1$ is at $\sigma$, then
\begin{equation}
  (P_\sigma)_{rk}
  =
  y_{r-k\bmod D}.
\end{equation}
Thus a value query to $P_\sigma$ can be simulated by one query to $y$.

Applying \Cref{lem:circuit-to-matrix} with $L=1$ embeds this one-gate circuit into a QLSP instance. If $\alpha$ is a fixed constant, then the probability of measuring the useful clock interval is also a fixed constant. Therefore an algorithm that outputs an $\eps$-close approximation to the QLSP solution state, for sufficiently small constant $\eps$, would allow us to recover $\sigma$ with bounded probability. This gives the basic $\Omega(\sqrt d)$ sparsity lower bound, up to the harmless additive shift between $D$ and the final sparsity parameter $d$.

\subsubsection{Condition number and precision from parity}
\label{subsec:warmup-kappa-epsilon}

We now show how the same clock construction gives the dependence $\kappa\log(1/\eps)$, even for constant sparsity. Let
\begin{equation}
  T
  =
  \left\lfloor
  \frac{\kappa\ln(1/\eps)}{100}
  \right\rfloor.
\end{equation}
Given oracle access to a bit string $z\in\{0,1\}^T$, consider the one-qubit circuit that applies $I$ if $z_i=0$ and $X$, the NOT gate, if $z_i=1$ at the $i$th step. The final state is
\begin{equation}
  \ket{\psi_{\mathrm{out}}}
  =
  \ket{\XOR_T(z)}.
\end{equation}
The circuit has length $T$, so in the notation of \Cref{lem:circuit-to-matrix} we have $L=T$. By this lemma, the useful clock probability is at least
\begin{equation}
  p
  \ge
  \frac{\alpha^{4T}}3,
  \qquad
  \alpha=1-\frac2\kappa.
\end{equation}
For sufficiently large $\kappa$,
\begin{equation}
  \alpha^{4T}
  =
  \left(1-\frac2\kappa\right)^{4T}
  \ge
  \eps^{4/25}.
\end{equation}
Thus the useful clock probability is much larger than $\eps$.

If an exact QLSP solver were available, we could measure the clock and, when the clock lies in the middle interval, read out the parity. If the clock is not in the middle interval, we output a uniformly random bit. This gives a parity algorithm with success probability $1/2+p/2$. If the QLSP solver produces a state within trace distance $\eps$ of the solution $\ket{x}$, the success probability of this two-outcome procedure changes by at most $\eps$. Since $p\gg\eps$, we still obtain an algorithm for parity with success probability strictly larger than $1/2$.
With a flagged solver, outputting a uniformly random bit on failure reduces this positive advantage by at most a factor of two.

Computing parity of $T$ input bits with any success probability strictly larger than $1/2$ requires $\Omega(T)$ queries~\cite{FGGS98,BBC01}. Hence any QLSP solver in this regime must make
\begin{equation}
  \Omega(T)
  =
  \Omega\left(\kappa\log(1/\eps)\right)
\end{equation}
queries.

\subsubsection{Toward the joint lower bound}
\label{subsec:warmup-product}

Our result combines these two warm-ups. We use one search instance to obtain a Boolean bit, and then take the parity of $R$ such bits. The search part contributes $\sqrt D$, while the parity length
\begin{equation}
  R
  =
  \Theta\left(\kappa\log(1/\eps)\right)
\end{equation}
is chosen so that the clock signal remains larger than the allowed trace-distance error.

There is one additional issue. The probability of the useful clock event decays exponentially in $R/\kappa$, so the Boolean algorithm obtained from a QLSP solver may have only a small advantage over random guessing. To turn such a small advantage into a query lower bound, we use the quantum XOR lemma of Lee and Roland~\cite{LR12}. This is the Boolean hardness input introduced next.

\subsection{Boolean source of hardness}
\label{sec:boolean-hardness}

We now define the Boolean promise problem embedded into the QLSP instance. Let $D$ be an even positive integer. The problem $\LR_D$ is defined on strings $y\in\{0,1\}^D$ of Hamming weight exactly one. If the unique $1$ lies in the left half of the string, then $\LR_D(y)=0$; if it lies in the right half, then $\LR_D(y)=1$. Equivalently,
\begin{equation}
  \LR_D(y)
  =
  \begin{cases}
    0, & \text{if the unique marked position is in } \{0,\ldots,D/2-1\},\\
    1, & \text{if the unique marked position is in } \{D/2,\ldots,D-1\}.
  \end{cases}
\end{equation}
The standard quantum lower bound for unstructured search~\cite{BBBV97} implies
\begin{equation}
  Q(\LR_D)
  =
  \Omega(\sqrt D).
  \label{eq:LR-lower-bound}
\end{equation}

For $R\ge1$, define the composed promise problem
\begin{equation}
  F_{R,D}
  \defeq
  \XOR_R\circ\LR_D.
\end{equation}
An input is a tuple $y=(y^1,\ldots,y^R)$, where each $y^i\in\{0,1\}^D$ has Hamming weight one, and
\begin{equation}
  F_{R,D}(y^1,\ldots,y^R)
  =
  \XOR_R\left(\LR_D(y^1),\ldots,\LR_D(y^R)\right).
  \label{eq:FRD-definition}
\end{equation}

We need the hardness of $F_{R,D}$ even when the algorithm is allowed error close to $1/2$. We use the quantum XOR lemma of Lee and Roland~\cite{LR12}.

\begin{lemma}[Small-advantage lower bound for $F_{R,D}$]
\label{lem:xor-lower-bound}
For every fixed constant $0<c<1/2$,
\begin{equation}
  Q_{1/2-c^R}(F_{R,D})
  =
  \Omega(R\sqrt D),
  \label{eq:FRD-small-advantage}
\end{equation}
where $Q_\eta(f)$ denotes the query complexity of computing $f$ with error at most $\eta$.
\end{lemma}

\begin{proof}
The Lee--Roland XOR lemma implies that, for every partial Boolean function $f$ and every fixed $0<c<1/2$,
\begin{equation}
  Q_{1/2-c^R}(\XOR_R\circ f)
  =
  \Omega(RQ(f)).
\end{equation}
Applying this with $f=\LR_D$ and using \cref{eq:LR-lower-bound} proves the claim.
\end{proof}

\subsection{The circuit for the composed search problem}

We now build the circuit that will be embedded into the QLSP. Let $D$ be a power of two. For a valid input $y^i\in\{0,1\}^D$ to $\LR_D$, let $\sigma_i$ be the unique marked position. Define the cyclic shift
\begin{equation}
  P_{\sigma_i}\ket j
  =
  \ket{j+\sigma_i\bmod D}.
\end{equation}
Its matrix entries are
\begin{equation}
  (P_{\sigma_i})_{rk}
  =
  y^i_{r-k\bmod D},
  \qquad
  (P_{\sigma_i}^{\dagger})_{rk}
  =
  y^i_{k-r\bmod D}.
  \label{eq:shift-entry-formulas}
\end{equation}
Thus a value query to $P_{\sigma_i}$ or $P_{\sigma_i}^{\dagger}$ can be simulated using one query to $y^i$.

The data Hilbert space is $\mathcal H_{\mathrm{data}} = \mathbb C^D\otimes\mathbb C^2$, where the first register is an address register and the second is an answer qubit. For $j\in\{0,\ldots,D-1\}$, define
\begin{equation}
  \operatorname{msb}(j)
  =
  \begin{cases}
    0, & 0\le j<D/2,\\
    1, & D/2\le j<D.
  \end{cases}
\end{equation}
Let $B$ be the fixed unitary
\begin{equation}
  B\ket j\ket a
  =
  \ket j\ket{a\oplus\operatorname{msb}(j)}.
  \label{eq:B-gate}
\end{equation}

For $y=(y^1,\ldots,y^R)$, define the circuit
\begin{equation}
  C_y
  =
  \prod_{i=1}^{R}
  \left[
    (P_{\sigma_i}^{\dagger}\otimes I_2)
    B
    (P_{\sigma_i}\otimes I_2)
  \right].
  \label{eq:exact-composed-circuit}
\end{equation}

\begin{lemma}[Exact circuit for $F_{R,D}$]
\label{lem:exact-circuit}
The circuit $C_y$ has length $3R$ and satisfies
\begin{equation}
  C_y\ket0\ket0
  =
  \ket0\ket{F_{R,D}(y)}.
\end{equation}
\end{lemma}
\begin{proof}
    For one block, starting from $\ket0\ket a$,
\begin{equation}
  (P_{\sigma_i}\otimes I_2)\ket0\ket a
  =
  \ket{\sigma_i}\ket a.
\end{equation}
Then
\begin{equation}
  B\ket{\sigma_i}\ket a
  =
  \ket{\sigma_i}\ket{a\oplus\operatorname{msb}(\sigma_i)}
  =
  \ket{\sigma_i}\ket{a\oplus\LR_D(y^i)}.
\end{equation}
Finally,
\begin{equation}
  (P_{\sigma_i}^{\dagger}\otimes I_2)
  \ket{\sigma_i}\ket{a\oplus\LR_D(y^i)}
  =
  \ket0\ket{a\oplus\LR_D(y^i)}.
\end{equation}
Repeating this for $i=1,\ldots,R$ leaves the address register in $\ket0$ and the answer qubit in the XOR of the $R$ left-right bits. Therefore
\begin{equation}
  C_y\ket0\ket0
  =
  \ket0\ket{F_{R,D}(y)}.
\end{equation}
\end{proof}

\subsection{Lower bound with joint complexity}

We now have all the ingredients to prove the lower bound in \Cref{thm:main_lower_bound}. 
Let $d$ be sufficiently large, and let $D$ be the largest power of two satisfying
\begin{equation}
  D+1\le d.
  \label{eq:D-choice}
\end{equation}
Then $D=\Theta(d)$. Let
\begin{equation}
  R
  =
  \left\lfloor
  \frac{\kappa\ln(1/\eps)}{100}
  \right\rfloor.
  \label{eq:R-choice}
\end{equation}
Consider the regime in which $R\ge1$ and
\begin{equation}
  R
  \ge
  \frac{\kappa\ln(1/\eps)}{200}.
  \label{eq:R-lower-bound}
\end{equation}

Let $y=(y^1,\ldots,y^R)$ be an input to $F_{R,D}$. For each block $y^i$, let $\sigma_i$ be the unique marked position, and define the circuit $C_y$ from \cref{eq:exact-composed-circuit}. By \Cref{lem:exact-circuit}, this circuit has length $L=3R$, and satisfies $C_y\ket0\ket0 = \ket0\ket{F_{R,D}(y)}$.

Apply the circuit-to-matrix construction to $C_y$. This gives
\begin{equation}
  A_y
  =
  \frac12\left(I-\alpha U_y\right),
  \qquad
  \alpha
  =
  1-\frac2\kappa,
  \label{eq:Ay-definition}
\end{equation}
with input state $\ket b = \ket1_{\mathrm{clock}}\ket0_{\mathrm{address}}\ket0_{\mathrm{answer}}$. By \Cref{lem:circuit-to-matrix}, $\norm{A_y}\le1$ and $\norm{A_y^{-1}}\le\kappa$.
Row and column location lists consist of the diagonal position and $D$ possible gate-entry positions, padding with zeros as needed; they are independent of $y$. By \cref{eq:shift-entry-formulas}, each value query costs $O(1)$ Boolean queries.

Suppose there is a QLSP algorithm using $q$ queries that solves every promised instance of this form. Conditioned on its success flag, it produces a state $\rho_y$ satisfying
\begin{equation}
  \Dtr\left(\rho_y,\ket{x_y}\!\bra{x_y}\right)
  \le
  \eps, \qquad \textrm{where} \quad \ket{x_y} = \frac{A_y^{-1}\ket{b}}{\norm{A_y^{-1}\ket{b}}}.
  \label{eq:rho-close}
\end{equation}
We convert the solver's output into a Boolean answer as follows. If the flag indicates failure, output a uniformly random bit. Otherwise, measure the clock register and the answer qubit of $\rho_y$. If the clock lies in the middle interval $\{L+1,\ldots,2L\}$, then output the measured answer qubit. Otherwise output a uniformly random bit.

First analyze this procedure on the exact state $\ket {x_y}$. By \Cref{lem:circuit-to-matrix}, the probability of the middle clock event is at least
\begin{equation}
  p
  \ge
  \frac{\alpha^{4L}}3
  =
  \frac{\alpha^{12R}}3.
  \label{eq:p-alpha-12R}
\end{equation}
For sufficiently large $\kappa$,
\begin{equation}
  \ln\left(1-\frac2\kappa\right)
  \ge
  -\frac4\kappa.
\end{equation}
Using $R\le\kappa\ln(1/\eps)/100$, we obtain
\begin{equation}
  \alpha^{12R}
  =
  \left(1-\frac2\kappa\right)^{12R}
  \ge
  \exp\left(-\frac{48R}{\kappa}\right)
  \ge
  \eps^{12/25}.
\end{equation}
Since $12/25<1/2$ and $0<\eps<1$,
\begin{equation}
  \eps^{12/25}
  \ge
  \sqrt\eps.
\end{equation}
Thus
\begin{equation}
  p
  \ge
  \frac{\sqrt\eps}{3}.
  \label{eq:p-sqrt-eps}
\end{equation}

Conditioned on the middle clock event, the data register is $C_y\ket0\ket0$, whose answer qubit is exactly $F_{R,D}(y)$. Therefore, on the exact state, the procedure succeeds with probability at least
\begin{equation}
  p\cdot1+(1-p)\frac12
  =
  \frac12+\frac p2
  \ge
  \frac12+\frac{\sqrt\eps}{6}.
  \label{eq:ideal-success}
\end{equation}
Replacing the exact state by $\rho_y$ changes the success probability of this two-outcome measurement by at most the trace distance in \cref{eq:rho-close}. Hence
\begin{equation}
  \Pr[\mathrm{correct}\mid\text{QLSP success}]
  \ge
  \frac12+\frac{\sqrt\eps}{6}-\eps.
\end{equation}
Since the solver succeeds with probability at least $1/2$ and we output a random bit on failure, for sufficiently small $\eps$ the overall success probability is at least
\begin{equation}
  \frac12+\frac{\sqrt\eps}{20}.
  \label{eq:boolean-advantage}
\end{equation}

Choose $c=e^{-100}$. Using \cref{eq:R-lower-bound}, $c^R = e^{-100R} \le e^{-(\kappa/2)\ln(1/\eps)} = \eps^{\kappa/2}$. For sufficiently large $\kappa$ and sufficiently small $\eps$,
\begin{equation}
  c^R
  \le
  \frac{\sqrt\eps}{20}.
\end{equation}
Thus the simulated Boolean algorithm computes $F_{R,D}$ with error at most $1/2-c^R$. By \Cref{lem:xor-lower-bound},
\begin{equation}
  O(q)
  \ge
  Q_{1/2-c^R}(F_{R,D})
  =
  \Omega(R\sqrt D).
\end{equation}
Finally, by \cref{eq:D-choice,eq:R-choice},
\begin{equation}
  R=\Theta\!\left(\kappa\log(1/\eps)\right),
  \qquad
  D=\Theta(d).
\end{equation}
Therefore
\begin{equation}
  q
  =
  \Omega\!\left(\kappa\sqrt d\log(1/\eps)\right).
\end{equation}
Although the matrix $A_y$ need not be Hermitian, the result extends to a Hermitian QLSP instance via a Hermitian dilation, without altering the relevant solution state. This completes the proof of \Cref{thm:main_lower_bound}.

Combining \Cref{thm:upper-bound,thm:main_lower_bound} gives the optimal query complexity
\begin{equation}
  \Theta(\kappa\sqrt d\log(1/\eps)),
\end{equation}
which proves our main result in \Cref{thm:main}.

\section{Block encoding of the original matrix}
\label{sec:block-encoding-A}

The main construction in \Cref{thm:sparse-reduction} is designed to produce $A^{-1}\ket b$, but it also
yields a constant-error block encoding of a Hermitian sparse matrix with
$O(\sqrt d)$ queries. 

We first combine \Cref{thm:sparse-reduction}
with block-encoding inversion
of Gily\'en et
al.~\cite{GSLW19} 
to implement inverses of well-conditioned
sparse matrices. For this, let us first define an approximate block encoding.  A unitary $U_K$ is an
$(\alpha,\eta)$ block encoding of $K$ if
\begin{equation}
  \norm{
    K-
    \alpha(\bra{0^a}\otimes\id)U_K(\ket{0^a}\otimes\id)
  }
  \leq\eta.
  \label{eq:approximate-block-encoding}
\end{equation}

\begin{lemma}[Block encoding an inverse]
\label{lem:sparse-inverse-block-encoding}
Let $A$ be an invertible Hermitian $s$-sparse matrix with
$\norm A\leq1$ and $\norm{A^{-1}}\leq\kappa=O(1)$.
Given sparse-matrix oracle access to $A$, for every $\eta>0$
one can construct an $(O(1),\eta)$ block encoding of $A^{-1}$ using
\begin{equation}
  O(\sqrt s\log(1/\eta)) \ \text{queries.}
  \label{eq:sparse-inverse-block-encoding-cost}
\end{equation}
\end{lemma}

\begin{proof}
\Cref{thm:sparse-reduction} gives an exact block encoding of a
contraction $M$ using $O(1)$ queries, with
\begin{equation}
  \norm{M^{-1}}\leq B\defeq5\kappa\sqrt s,
  \label{eq:appendix-main-theorem-output}
\end{equation}
and the bottom-right block of $M^{-1}$ is $\sqrt s\,A^{-1}$.

We now turn the block encoding of $M$ into one of a scaled $M^{-1}$.
The block-encoding inversion procedure of Gily\'en et
al.~\cite[Theorem~41 and Corollary~69]{GSLW19} uses the bound
$\norm{M^{-1}}\leq B$ to construct a unitary whose encoded block
approximates $\frac{3}{4B}M^{-1}$ to operator-norm error $\delta$, using
\begin{equation}
  O\left(
    B\log(1/\delta)
  \right)
  =
  O\left(
    \sqrt s\log(1/\delta)
  \right)
  \label{eq:appendix-qsvt-cost}
\end{equation}
uses of the block encoding of $M$. Selecting the bottom-right block of
the encoded matrix requires no further queries and gives an approximation to
\begin{equation}
  \frac{3}{4B}\sqrt s\,A^{-1}
  =
  \frac{3}{20\kappa}A^{-1}.
  \label{eq:appendix-cancellation}
\end{equation}
Thus the factor $\sqrt s$ in the inverse block cancels the sparsity
dependence of the inversion normalization. Taking
$\delta=3\eta/(20\kappa)$ gives an $(O(1),\eta)$ block encoding of
$A^{-1}$ using $O(\sqrt s\log(1/\eta))$ sparse-matrix oracle queries.
\end{proof}

\Cref{lem:sparse-inverse-block-encoding} lets us block encode inverses of
well-conditioned sparse matrices. We can then use it to block encode a
Hermitian sparse matrix $H$ itself, even when $H$ is singular or 
poorly conditioned. The idea is to realize $H$ as a block
of the inverse of a larger, well-conditioned sparse matrix.

\sqrtdblockencoding*

\begin{proof}
To apply \Cref{lem:sparse-inverse-block-encoding}, we seek a
well-conditioned sparse linear system whose solution contains $H\vec b$. Introduce
variables $\vec u,\vec v$ and impose $\vec v=\vec b$ and
$\vec u-H\vec v=0$. Equivalently,
\begin{equation}
  \underbrace{
  \begin{pmatrix}
    0&\id\\
    \id&-H
  \end{pmatrix}}_{=\,G_H}
  \begin{pmatrix}\vec u\\\vec v\end{pmatrix}
  =
  \begin{pmatrix}\vec b\\0\end{pmatrix}.
  \label{eq:GH-definition}
\end{equation}
The solution has $\vec u=H\vec b$, and indeed
\begin{equation}
  G_H^{-1}
  =
  \begin{pmatrix}
    H&\id\\
    \id&0
  \end{pmatrix}.
  \label{eq:GH-inverse}
\end{equation}
The matrix $G_H$ is Hermitian and $(d+1)$-sparse, with oracles
implementable using $O(1)$ queries to those for $H$. Both $G_H$ and
$G_H^{-1}$ have norm at most $1+\norm H\leq2$. Hence
$K_H\defeq G_H/2$ is a contraction with $\norm{K_H^{-1}}\leq4$.
Applying \Cref{lem:sparse-inverse-block-encoding} to $K_H$
therefore uses $O(\sqrt d\log(1/\eta))$ queries. Since
$K_H^{-1}=2G_H^{-1}$ has upper-left block $2H$, selecting this block
and halving the normalization gives the desired $(O(1),\eta)$ block
encoding of $H$.
\end{proof}

We also obtain the optimal bounded-error query complexity for implementing
a unitary from its entries, answering the question of Berry and
Childs~\cite{BC12}.

\blackboxunitary*

\begin{proof}
The Hermitian dilation
\begin{equation}
  H_U=\begin{pmatrix}0&U\\U^\dagger&0\end{pmatrix}
  \qquad\text{satisfies}\qquad
  H_U^{-1}=H_U,\quad \norm{H_U}=1.
  \label{eq:unitary-dilation}
\end{equation}
It is $N$-sparse, with location lists given by the opposite block and
entry oracles implementable using $O(1)$ queries to those for $U$.
\Cref{lem:sparse-inverse-block-encoding}, with $\kappa=1$, therefore
gives an $(O(1),\eta)$ block encoding of $H_U$ using
$O(\sqrt N\log(1/\eta))$ queries. Selecting its upper-right block
gives the same guarantee for $U$.

Robust oblivious amplitude amplification~\cite[Lemma~6]{BCK15}
uses $O(1)$ calls to this block encoding and its inverse, enlarging the
constant normalization if necessary, to obtain an encoded block within
$O(\eta)$ of $U$. Since $U$ is unitary, this implements $U$ with
state error $O(\sqrt\eta)$. Taking $\eta=c\eps^2$ for a sufficiently
small constant $c>0$ proves the upper bound. The $\Omega(\sqrt N)$
lower bound follows from unstructured search~\cite{BC12}.
\end{proof}

\bibliographystyle{alphaurl}
\bibliography{citations}

\end{document}